%% file: arxiv.tex
\documentclass[10pt, conference, letterpaper]{IEEEtran}
\usepackage{amsmath,amsfonts}
\usepackage{amssymb}
\pdfoutput=1
\usepackage{algorithmic}
\usepackage{algorithm}
\usepackage{array}
\usepackage[caption=false,font=normalsize,labelfont=sf,textfont=sf]{subfig}
\usepackage{textcomp}
\usepackage{stfloats}
\usepackage{url}
\usepackage{verbatim}
\usepackage{graphicx}
\usepackage{epstopdf}
\usepackage{etoolbox}
\usepackage[colorlinks,
            linkcolor=red,
            anchorcolor=blue,
            citecolor=black
            ]{hyperref}
\usepackage{cite}
\usepackage{bm}
\begin{document}
\pagestyle{plain}
\newtheorem{theorem}{Theorem}
\newtheorem{proof}{Proof}

\title{Beamforming Design for Integrated Sensing and Communication with Extended Target}

\makeatletter
\newcommand{\linebreakand}{%
  \end{@IEEEauthorhalign}
  \hfill\mbox{}\par
  \mbox{}\hfill\begin{@IEEEauthorhalign}
}
\makeatother
\author{Yiqiu Wang, Meixia Tao, and Shu Sun\\
Department of Electronic Engineering, Shanghai Jiao Tong University, Shanghai, China \\
Emails: \{wyq18962080590, mxtao, shusun\}@sjtu.edu.cn
}

\maketitle

\thispagestyle{plain}

\begin{abstract}
This paper studies transmit beamforming design in an integrated sensing and communication (ISAC) system, where a base station sends symbols to perform downlink multi-user communication and sense an extended target simultaneously. We first model the extended target contour with truncated Fourier series. By considering echo signals as reflections from the valid elements on the target contour, a novel Cramér-Rao bound (CRB) on the direction estimation of extended target is derived. We then formulate the transmit beamforming design as an optimization problem by minimizing the CRB of radar sensing, and satisfying a minimum signal-to-interference-plus-noise ratio requirement for each communication user, as well as a 3-dB beam coverage requirement tailored for the extended sensing target under a total transmit power constraint. In view of the non-convexity of the above problem, we employ semidefinite relaxation (SDR) technique for convex relaxation, followed by a rank-one extraction scheme for non-tight relaxation circumstances. Numerical results show that the proposed SDR beamforming scheme outperforms benchmark beampattern design methods with lower CRBs for the circumstances considered.
\end{abstract}
\let\thefootnote\relax\footnotetext{This work is supported by the NSF of China under Grant 62125108 and Grant 62271310, and by the Fundamental Research Funds for the Central Universities of China.}
\section{Introduction}
\IEEEPARstart{C}{ommunication} and sensing share similar development trends, including higher frequency bands, larger antenna arrays and hardware miniaturization. Integrated sensing and communication (ISAC) has become one of the key usage scenarios of 6G. Among various ISAC systems, the combination of radar sensing and communication is a common form. However, the fundamental goals to be achieved by radar and communication systems are completely different. Traditional radar designs focus on extracting target information from the reflected echo signals, whereas communication aims to transmit information accurately at minimal cost. As such, a major design issue of ISAC is how to accommodate the diverse design objectives of communication and radar sensing.

Recently there have been many research efforts devoted to the beamforming design for joint radar sensing and communication. One typical approach is to take into account both communication and radar sensing metrics in the beamforming optimization problem, thereby achieving performance tradeoff for communication and sensing. While data rate and bit error rate are commonly employed to evaluate the communication performance, Cramér-Rao bound (CRB) defines the lower bound of the variance for any unbiased estimator, and thus serves as a widely used sensing performance metric.

Deriving closed-form CRBs, especially in the extended target cases, has always been the focus of ISAC literature. In particular, by assuming there exists line-of-sight (LoS) between an extended sensing target (EST) and a base station (BS), the authors in \cite{ref4} estimate the whole response matrix and obtain a closed-form CRB expression as a function of the precoding matrix. The authors in \cite{ref5} consider the non-line-of-sight (NLoS) ISAC scenario and derive the response matrix CRB for reconfigurable intelligent surface aided sensing. Although further information, such as the target range or direction, can be extracted from the response matrix with sophisticated signal processing algorithms, it is awfully unintuitive and unprecise to use CRB on the intermediate response matrix to represent estimation accuracy of the true desired parameters. In regard of the problem above, the authors in \cite{ref6} propose an elaborate CRB analysis on the estimation for central range, direction, and orientation of an extended target with known and unknown contours. Nevertheless, this work merely considers the multiple-input single-output scenario. The explicit CRB expression is also too complicated to construct a solvable optimization problem, especially when the communication function is integrated along with the sensing task. To the best of our knowledge, there is limited work studying the CRB optimization of EST in the multi-user multiple-output multiple-input (MU-MIMO) ISAC scenario.

In this paper, we propose a transmit beamforming design framework for MU-MIMO ISAC systems, with a specific emphasis on the optimization of CRB for extended target sensing. The main contributions of this work are summerized as follows:

\begin{itemize}
    \item First, we derive a novel closed-form CRB expression for the direction estimation performance of the extended target. This expression reveals the main dependence of CRB upon sensing path loss, noise ratio, number of transceiver antennas, and signal covariance matrix.
    \item Second, we formulate a CRB minimization problem by optimizing the transmit beamforming subject to communication-specific constraints and beam coverage constraints. Due to the non-convex fraction and quadratic structure of the CRB, we propose a semidefinite relaxation (SDR) based algorithm to obtain sub-optimal solutions. This algorithm also consists of a novel procedure to extract rank-one beamformers for non-tight relaxation scenarios.
    \item Compared with existing designs, numerical results show that our proposed SDR design effectively boosts sensing performance in terms of CRB.
\end{itemize}

\textit{Notation:} $\left[\cdot\right]^T$, $\left[\cdot\right]^H$, $\left[\cdot\right]^*$ denote, respectively, the transpose, Hermitian transpose, and conjugate of a matrix; $\mathbb{E}\left[\cdot\right]$ denotes the mean of variables; $\mathbf{0}_{m \times n}$ and $\mathbf{I}_{m}$ denote an ${m \times n}$ matrix with all zero elements, and an ${m \times m}$ identity matrix, respectively; $\mathcal{R\left(\cdot\right)}$ and $\mathcal{I\left(\cdot\right)}$ respectively denote the real and imaginary part of a complex number; $\mathcal{CN}\left(\mathbf{0}_{m\times 1},\sigma^{2}\mathbf{I}_m\right)$ denotes the probability density of an ${m\times 1}$ circularly symmetric complex Gaussian vector with zero mean and covariance matrix $\sigma^{2}\mathbf{I}_m$; $\mathbb{R}^{m\times n}$ and $\mathbb{C}^{m\times n}$ denote a matrix with ${m\times n}$ real elements and ${m\times n}$ complex elements, respectively; $\Delta _{{{\boldsymbol{\theta} }_{1}}}^{{\boldsymbol{\theta }_{2}}}\left[ \cdot \right]$ denotes the second derivative over ${\boldsymbol{\theta} }_{1}$ and ${\boldsymbol{\theta} }_{2}$; $\mathbf{A} \geq 0$ denotes that matrix $\mathbf{A}$ is semi-definite; $\Vert\cdot\Vert$ denotes the $l_2$ norm.

\section{System Model}
We consider a downlink ISAC system. A monostatic MIMO radar and an MU-MIMO communication transmitter are integrated inside a BS, which is equipped with {$N_t$} transmit antennas and {$N_r$} receive antennas. The BS sends wireless signals to perform radar sensing tasks for one extended target and to communicate with $C$ single-antenna communication users (CUs) simultaneously. The CUs are capable of decoding communication messages based on their own received signals. At the same time, the BS receives echo signals reflected from the surface of the extended target, from which the unknown parameters of the extended target, including the central range, direction, and orientation, can be extracted with specific methods. To guarantee the feasibility of basic radar sensing and communication function, we assume $C \leq N_t \leq N_r$ throughout the paper.

The transmitted signal is expressed as

\begin{equation}
\label{eq_l1}
\mathbf{x}(t) = \mathbf{W}_c \mathbf{c}(t),
\end{equation}
where $\mathbf{W}_c = \left[\mathbf{w}_1,...,\mathbf{w}_C \right] \in \mathbb{C}^{N_{t}\times C}$ and $\mathbf{c}(t) = \left[{c}_{1}(t),...,{c}_{C}(t) \right]^{T}$ are the precoding matrix and communication symbols of $C$ CUs respectively. Here, $\mathbf{w}_c$ and $c_{c}(t)$ are respectively the beamforming vector and transmitted communication symbol of the $c$-th CU. 

We assume that communication signals are temporally-white and wide-sense stationary stochastic process with zero-mean. Further, communication symbols intended for different CUs are uncorrelated with each other, i.e. $\mathbb{E} \left[\mathbf{c}(t) \mathbf{c}^H (t) \right] = \mathbf{I}_{C}$. The covariance matrix of transmitted signals is presented as
\begin{equation}
\label{eq_l2}
\mathbf{R}_x = \mathbb{E} \left[\mathbf{x}(t) \mathbf{x}^H (t) \right] = \mathbf{W}_c \mathbf{W}_c^H.
\end{equation}

The received signal of the $c$-th CU at time $t$ is presented as
\begin{equation}%会产生编号
\label{eq_l3}
{y}_{c}(t) = \mathbf{h}_{c}^{H} \mathbf{x}(t) + {z}_{c}(t) = \mathbf{h}_{c}^{H} \sum\nolimits_{c=1}^C \mathbf{w}_{c}(t)c_{c}(t) + {z}_{c}(t),
\end{equation}
where ${y}_c(t)$ is the received signal, ${z}_c(t)$ is the additive white Gaussian noise (AWGN) with zero-mean and $\sigma_c^2$ variance, and $\mathbf{h}_c \in \mathbb{C}^{N_{t}}$ is the narrowband communication channel vector following the Saleh-Valenzuela model.

To guarantee the communication quality at each CU, the precoder should be designed to achieve a certain level of signal-to-interference-plus-noise ratio (SINR) for all users. For the $c$-th CU, the SINR is defined as

\begin{equation}
\label{eq_l4}
{{\gamma }_{c}}={{{\left| \mathbf{h}_{c}^{H}{{\mathbf{w}}_{c}} \right|}^{2}}}/\left({\sum\nolimits_{i=1,i\ne c}^{C}{{{\left| \mathbf{h}_{c}^{H}{{\mathbf{w}}_{i}} \right|}^{2}}+\sigma _{c}^{2}}}\right).
\end{equation}

\section{Extended Target Sensing}
\subsection{Extended Target Contour Model}

\begin{figure}[!t]
\centering
\includegraphics[width=3.5in]{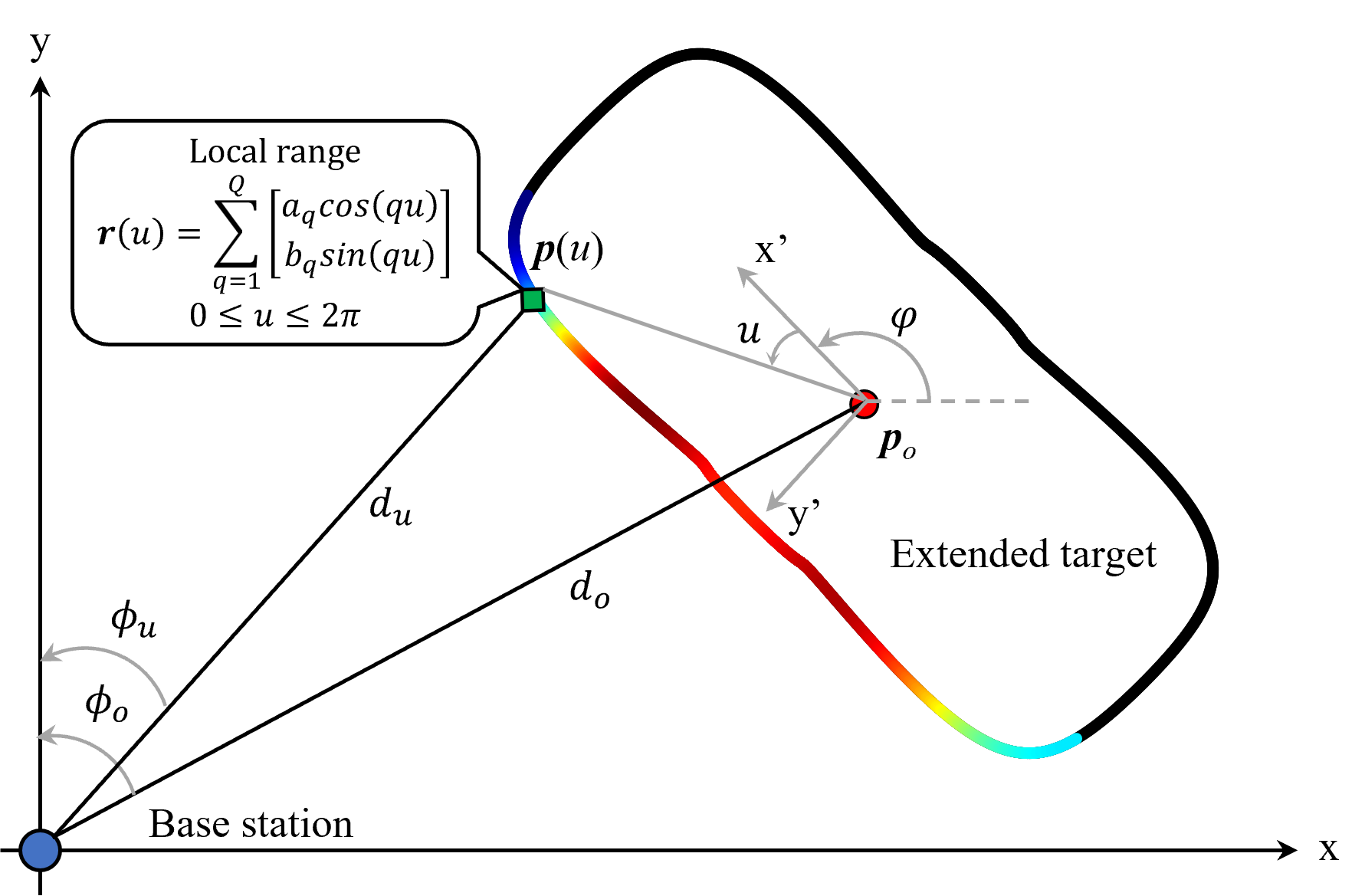}
\caption{Bird-eye view of an extended sensing target with TFS contour. The different colors on the contour with LoS link to the BS indicate strength of reflected echo signals, from red (maximum) to blue (minimum). The NLoS contour are plotted with black curve which provides no reflection.}
\label{fig1}
\end{figure}

In our work, the effective echo signals received at the BS are assumed to be reflected from the target contour elements with LoS link to the BS. We build on the truncated Fourier series (TFS) model used in \cite{ref6} to derive an analytical EST contour.

As presented in Fig. $\ref{fig1}$, the extended target is located at the origin of the local coordinate system where the $+x$ axis is aligned with the heading of the target. The target orientation is defined as the angle $\varphi$ from the global $+x$ axis to the local $+x$ axis. The local coordinate of each contour element is represented by the TFS with $2Q$ coefficients as a function of the local direction $u$. As such, the contour can be generated for $0\le u\le 2\pi$ in the local coordinate, by

\begin{equation}
\label{eq_l5}
\mathbf{r}\left( u \right)=\sum\limits_{q=1}^{Q}{\left[ \begin{matrix}
    {{a}_{q}}\cos \left( qu \right) \\ 
    {{b}_{q}}\sin \left( qu \right) \\ 
\end{matrix} \right]}=\left[ \begin{matrix}
   {{\boldsymbol{\sigma }}^{T}}\mathbf{m}  \\
   {{\boldsymbol{\varsigma }}^{T}}\mathbf{n}  \\
\end{matrix} \right],
\end{equation}
where $\boldsymbol{\sigma }={{\left[ \cos \left( u \right),...,\cos \left( Qu \right) \right]}^{T}}$ is the cosine harmonics, $\boldsymbol{\varsigma }={{\left[ \sin \left( u \right),...,\sin \left( Qu \right) \right]}^{T}}$ is the sine harmonics, $\mathbf{m}={{\left[ {{a}_{1}},...,{{a}_{Q}} \right]}^{T}}$is the TFS cosine coefficient and $\mathbf{n}={{\left[ {{b}_{1}},...,{{b}_{Q}} \right]}^{T}}$ is the TFS sine coefficient.

Assume the BS is located at the origin of the global coordinate system and $+x$ axis is aligned with the antenna orientation. The displacement of the target at a given position $\mathbf{p}_o$ can be presented as
\begin{equation}
\label{eq_l6}
{{\mathbf{p}}_{o}}={{\left[ {{d}_{o}}\sin {{\phi }_{o}}\text{   }{{d}_{o}}\cos {{\phi }_{o}} \right]}^{T}},
\end{equation}
where ${{d}_{o}}$ represents the range between the BS and the target center, ${{\phi }_{o}}$ is the target direction. For a specific element on the contour, its global displacement can be expressed as
\begin{equation}
\label{eq_l7}
\mathbf{p}\left( u \right)={{\mathbf{p}}_{o}}+\mathbf{V}\mathbf{r}\left( u \right),
\end{equation}
where $\mathbf{V=}\left[ \begin{matrix} \cos \varphi  & -\sin \varphi \\ \sin \varphi  & \cos \varphi \\ \end{matrix} \right]$ is the spin matrix at orientation $\varphi$ with regard to the $+x$ axis in the global coordinate system. Finally, we obtain the complete target contour, defined as $\mathcal{C}=\left\{ \mathbf{p}\left( u \right):0\le u\le 2\pi \right\}$.

\subsection{Received Sensing Signal Model}
We consider a full-duplex system, where the BS captures the echo signal reflected from the LoS contour of EST. We define $\mathcal{C}_{LoS}$ and ${u}_{LoS}$ as the target LoS contour and local direction of LoS contour elements respectively. While the target contour is a continuous curve on the 2D plane, we can split $\mathcal{C}_{LoS}$ into $K$ disjoint subsections with an angular interval of $\Delta u={{u}_{LoS}}/K$, then we have ${\mathcal{C}_{LoS}}=\bigcup _{k=1}^{K}{\mathcal{C}_{k}}$ where ${\mathcal{C}_{k}}=\left\{ \mathbf{p}\left( {{u}_{k}} \right),u_{LoS}^{lower}+\left( k-1 \right)\Delta u\le {{u}_{k}}\le u_{LoS}^{lower}+k\Delta u \right\}$. The local directions $u_{LoS}^{lower}\le {{u}_{1}}\le u_{LoS}^{lower}+\Delta u\le {{u}_{2}}\le ...\le {{u}_{K}}\le u_{LoS}^{upper}=u_{LoS}^{lower}+K\Delta u$ define a nonoverlapping partition of the local LoS angular range $\left[ u_{LoS}^{lower},u_{LoS}^{upper} \right]$. Consequently, the received echo signal at the BS can be expressed as
\begin{align}
\mathbf{e}\left( t \right)&=\int_{{\mathcal{C}_{LoS}}}{{{\mathbf{e}}_\mathbf{r}}\left( t \right)}\text{d}\mathbf{r}=\sum\limits_{k=1}^{K}{{\mathbf{e}}_{k}}\left( t \right) \nonumber \\
&=g\sum\limits_{k=1}^{K}{\sqrt{{l}_{k}}{{\alpha }_{k}}\mathbf{b}\left( {{\phi }_{k}} \right){{\mathbf{a}}^{H}}\left( {{\phi }_{k}} \right)\mathbf{x}\left( t-\frac{2{{d}_{k}}}{c} \right)},
\label{eq_l8}
\end{align}
where ${{\mathbf{e}}_{\mathbf{r}}}$ and ${{\mathbf{e}}_{{k}}}$ refer to the echo signals as a function of $\mathbf{r}$ and $\mathcal{C}_k$, ${{\alpha }_{k}}\sim \mathcal{CN}\left(0,1 \right)$, ${{\phi }_{k}}$, ${{d}_{k}}$ and $l_k$ refer to the complex radar cross section (RCS), global direction, range, and perimeter of $\mathcal{C}_k$, respectively, $g=1/d_{o}^{2}$ is the sensing path loss coefficient, $\mathbf{a}(\cdot)$ and $\mathbf{b}(\cdot)$ are the steering vectors of transmit and receive antennas, respectively.

After mixing with AWGN in the wireless channel, the received sensing signal at the BS is
\begin{equation}
\label{eq_l9}
{{\mathbf{y}}_{s}}\left( t \right)=\mathbf{e}\left( t \right)+{{\mathbf{z}}_{s}}\left( t \right),
\end{equation}
where ${{\mathbf{z}}_{s}}\left( t \right)\in {\mathbb{C}^{{{N}_{r}}}}$ is the AWGN vector with zero-mean and $\sigma _{s}^{2}$ variance for each element.

\subsection{CRB for Extended Target}
We define $\boldsymbol{\xi }={{\left[ {{\boldsymbol{\alpha }}^{T}}\text{ } g \text{ } {{\boldsymbol{\kappa}}^{T}} \right]}^{T}}$ as the unknown parameter set including three kinds of parameters. $\boldsymbol{\alpha }=\left[ \mathcal{R}\left( {{\alpha }_{1}},{{\alpha }_{2}},...,{{\alpha }_{K}} \right),\mathcal{I}\left( {{\alpha }_{1}},{{\alpha }_{2}},...,{{\alpha }_{K}} \right) \right]^{T}\in {\mathbb{R}^{2K}}$ is a nuisance RCS vector parameter which has random values for every observation, $g$ is the path loss coefficient as defined before, and $\boldsymbol{\kappa }={{\left[ {{d}_{o}}\text{ }{{\phi }_{o}}\text{ }\varphi \text{ }\mathbf{m}^T \text{ }\mathbf{n}^T \right]}^{T}}\in {\mathbb{R}^{\left( 2Q+3 \right)}}$ is the deterministic vector parameter of interest. As stated in Section III-B, the received echo signals are independent complex Gaussian vectors with ${{\mathbf{y}}_{s}}\left( t \right)\sim \mathcal{CN}\left( \mathbf{e}\left( t \right),\sigma _{s}^{2}{{\mathbf{I}}_{{{N}_{r}}}} \right)$. Within a certain observation period $t_s$, the log-likelihood function for estimating $\boldsymbol{\xi }$ from ${{\mathbf{y}}_{s}}\left( t \right)$ is presented as
\begin{align}
\log p\left( {{\mathbf{y}}_{s}}|\boldsymbol{\xi } \right)= &\frac{2}{{\sigma_s^2}}\mathcal{R} \int_{t_s}{\mathbf{y}_{s}^{H}\mathbf{e}\left( t \right)\text{d}t}-\frac{1}{{\sigma_s^2}}\int_{t_s}{{{\left\| \mathbf{e}\left( t \right) \right\|}^{2}}\text{d}t} \nonumber \\ 
& - \frac{1}{{\sigma_s^2}}\int_{t_s}{{{\left\| \mathbf{y}\left( t \right) \right\|}^{2}}\text{d}t} -t_{s}N_{r}\text{log}\left(\pi\sigma_s^{2}\right).
\label{eq_l10}
\end{align}

According to the definition of Fisher information matrix (FIM) \cite{ref6}, we obtain the FIM of all parameters as
\begin{align}
\mathbf{J}\left( \boldsymbol{\xi } \right)&=-\mathbb{E} \left[ \Delta _{\boldsymbol{\xi }}^{\boldsymbol{\xi }}\log p\left( {{\mathbf{y}}_{s}},\boldsymbol{\alpha |}g,\boldsymbol{\kappa } \right) \right] \nonumber \\ 
 & =-\mathbb{E} \left[ \Delta _{\boldsymbol{\xi }}^{\boldsymbol{\xi }}\log p\left( {{\mathbf{y}}_{s}}\boldsymbol{|\alpha },g,\boldsymbol{\kappa } \right) \right]-\mathbb{E} \left[ \Delta _{\boldsymbol{\xi }}^{\boldsymbol{\xi }}\log p\left( \boldsymbol{\alpha |}g,\boldsymbol{\kappa } \right) \right]\nonumber \\ 
 & =-\mathbb{E} \left[ \Delta _{\boldsymbol{\xi }}^{\boldsymbol{\xi }}\log p\left( {{\mathbf{y}}_{s}}\mathbf{|\xi } \right) \right]-\mathbb{E} \left[ \Delta _{\boldsymbol{\xi }}^{\boldsymbol{\xi }}\log p\left( \boldsymbol{\alpha } \right) \right],
 \label{eq_l11}
\end{align}
where $p\left( {{\mathbf{y}}_{s}},\boldsymbol{\alpha |}g,\boldsymbol{\kappa } \right)$ is the joint a posteriori probability density function of the echo signals. Note that the second equation in (\ref{eq_l11}) holds as a direct application of the Bayes theorem, while the third equation explores the fact that RCS $\boldsymbol{\alpha }$ is a random vector independent of $g$ and $\boldsymbol{\kappa }$.

Since the received signals $\mathbf{y}_s(t)$ obey Gaussian distribution, we can present the expected second derivative of the log-likelihood function as
\begin{equation}
\label{eq_l12}
-\mathbb{E} \left[ \Delta _{{{\boldsymbol{\theta} }_{1}}}^{{\boldsymbol{\theta }_{2}}}\log p\left( {{\mathbf{y}}_{s}}\mathbf{|\xi } \right) \right]=\frac{2}{\sigma _{s}^{2}} \mathcal{R}\int_{{{t}_{s}}}{\mathbb{E}\left[ \frac{\partial {{\mathbf{e}}^{H}}}{\partial {\boldsymbol{\theta }_{1}}}\frac{\partial \mathbf{e}}{\partial \boldsymbol{\theta}_{2}^{T}} \right]\text{d}t},
\end{equation}
where ${{\boldsymbol{\theta} }_{1}}$ and ${{\boldsymbol{\theta} }_{2}}$ are arbitrary variables. Once again, considering the random characteristic of RCS, we derive that $\mathbb{E}\left[ \Delta _{g}^{\boldsymbol{\alpha }}\log p\left( \boldsymbol{\alpha } \right) \right]={{\mathbf{0}}_{2K\times 1}}$, $\mathbb{E}\left[ \Delta _{\boldsymbol{\kappa }}^{\boldsymbol{\alpha }}\log p\left( \boldsymbol{\alpha } \right) \right]={{\mathbf{0}}_{2K\times (2Q+3)}}$. Further, since $\mathbb{E}\left[ {{{\alpha }}_{k}} \right]=0$, we have $\mathbb{E}\left[ \Delta _{g}^{\boldsymbol{\alpha }}\log p\left( {{\mathbf{y}}_{s}}|\boldsymbol{\xi } \right) \right]={{\mathbf{0}}_{2K\times 1}}$ and $\mathbb{E}\left[ \Delta _{\boldsymbol{\kappa }}^{\boldsymbol{\alpha }}\log p\left( {{\mathbf{y}}_{s}}|\boldsymbol{\xi } \right) \right]={{\mathbf{0}}_{2K\times (2Q+3)}}$.

Combining the above properties, the FIM can be split as
\begin{equation}
\label{eq_l13}
\mathbf{J}\left( \boldsymbol{\xi } \right)=\left[ \begin{matrix}
   {{\mathbf{I}}_{\boldsymbol{\alpha }}} & {{\mathbf{0}}_{2K\times \left( 2Q+4 \right)}}  \\
   {{\mathbf{0}}_{\left( 2Q+4 \right)\times 2K}} & {{\mathbf{I}}_{\left( \boldsymbol{\kappa },g \right)}}  \\
\end{matrix} \right],
\end{equation}
\begin{equation}
\label{eq_l14}
{{\mathbf{I}}_{\left( \boldsymbol{\kappa },g \right)}}=\left[ \begin{matrix}
   {{i}_{g}} & \mathbf{i}_{\boldsymbol{\kappa },g}^{T}  \\
   {{\mathbf{i}}_{\boldsymbol{\kappa },g}} & {{\mathbf{I}}_{\boldsymbol{\kappa }}}  \\
\end{matrix} \right],
\end{equation}
where ${{\mathbf{I}}_{\boldsymbol{\alpha }}}\in {\mathbb{R}^{2K\times 2K}}$, ${{\mathbf{I}}_{\boldsymbol{\kappa }}}\in {\mathbb{R}^{\left( 2Q+3 \right)\times \left( 2Q+3 \right)}}$ and ${{\mathbf{I}}_{\left( \boldsymbol{\kappa },g \right)}}\in {\mathbb{R}^{\left( 2Q+4 \right)\times \left( 2Q+4 \right)}}$ are FIMs of the RCS parameter $\boldsymbol{\alpha }$, the deterministic vector $\boldsymbol{\kappa }$ and the combination of $\left( \boldsymbol{\kappa },g \right)$, respectively. ${{i}_{{g}}}$ is the Fisher information scalar related to the path loss $g$, ${{\mathbf{i}}_{\boldsymbol{\kappa },g}}\in {\mathbb{R}^{2Q+3}}$ refers to the Fisher vector regarding the path loss $g$ and parameters of interest $\boldsymbol{\kappa }$.

It should be noted that while there are a total of $2Q+2K+4$ unknown parameters in $\boldsymbol{\xi }$, we are only interested in the specific estimation of $2Q+3$ parameters in $\boldsymbol{\kappa }$. Thus, the effective Fisher information matrix (EFIM) should be extracted from the overall FIM for further analysis. We commence by ignoring the FIM elements related to $\boldsymbol{\alpha }$ given its irrelevance of $\boldsymbol{\kappa }$ and $g$, and only focus on the lower-right matrix block ${{\mathbf{I}}_{\left( \boldsymbol{\kappa },g \right)}}$. Then, we invert ${{\mathbf{I}}_{\boldsymbol{\kappa }}}$ as the Schur’s complement of ${{i}_{g}}$ over ${{\mathbf{I}}_{\boldsymbol{\kappa },g}}$ and calculate the EFIM as
\begin{equation}
\label{eq_l15}
\mathbf{J}\left( \boldsymbol{\kappa } \right)={{\mathbf{I}}_{\boldsymbol{\kappa }}}-\frac{1}{{{i}_{g}}}{{\mathbf{i}}_{\boldsymbol{\kappa },g}}\mathbf{i}_{\boldsymbol{\kappa },g}^{T}.
\end{equation}

\begin{figure*}[t]
\begin{equation}
\label{eq_l16}
CRB\left( {{\phi }_{o}} \right)={{\left( \frac{2{{g}^{2}}{{N}_{r}t_s}}{\sigma _{s}^{2}} \right)}^{-1}}{{\left\{ \sum\limits_{k=1}^{K}{{{l}_{k}}\left[ {{Z}_{1}}\mathbf{a}_{k}^{H}{{\mathbf{R}}_{x}}{{\mathbf{a}}_{k}}+\mathbf{\dot{a}}_{k}^{H}{{\mathbf{R}}_{x}}{{{\mathbf{\dot{a}}}}_{k}}-\frac{{{\left( \mathbf{\dot{a}}_{k}^{H}{{\mathbf{R}}_{x}}{{\mathbf{a}}_{k}}+\mathbf{a}_{k}^{H}{{\mathbf{R}}_{x}}{{{\mathbf{\dot{a}}}}_{k}} \right)}^{2}}}{4\mathbf{a}_{k}^{H}{{\mathbf{R}}_{x}}{{\mathbf{a}}_{k}}} \right]} \right\}}^{-1}}.\\
\end{equation}
\end{figure*}

With prior information of the RCS parameter $\boldsymbol{\alpha }$, the simplified EFIM $\mathbf{J}\left( \boldsymbol{\kappa } \right)$ can be extracted from the original FIM $\mathbf{J}\left( \boldsymbol{\xi } \right)$. Nevertheless, we still need to calculate the matrix inverse of EFIM to obtain the final CRB matrix on $\boldsymbol{\kappa }$. It should be noted that EFIM $\mathbf{J}\left( \boldsymbol{\kappa } \right)$ has a dimension of $\left( 2Q+3 \right)\times \left( 2Q+3 \right)$, where the number of TFS parameters $Q$ is generally greater than eight for an acceptable representation of the contour. Hence, considering the matrix inversion operation, it is hard to obtain a closed-form CRB on a specific parameter directly from $\mathbf{J}\left( \boldsymbol{\kappa } \right)$, say the target direction, not to mentioned constructing a solvable problem aimed at optimizing the derived CRB. Consequently, we further define $\boldsymbol{\kappa =}\left[ {{\boldsymbol{\kappa }}_{1}^T},{{\boldsymbol{\kappa }}_{2}^T} \right]^T\text{, }{{\boldsymbol{\kappa }}_{1}}=\left[ {{d}_{o}},{{\phi }_{o}},\varphi  \right]^T\text{, }{{\boldsymbol{\kappa }}_{2}}=\left[ \mathbf{m}^T,\mathbf{n}^T \right]^T$, and make the following assumptions
\begin{equation}
{{\mathbf{I}}_{\boldsymbol{\kappa }}}=\left[ \begin{matrix}
   {{\mathbf{I}}_{{{\boldsymbol{\kappa }}_{1}}}} & \mathbf{I}_{{{\boldsymbol{\kappa }}_{2}},{{\boldsymbol{\kappa }}_{1}}}^{T}  \\
   {{\mathbf{I}}_{{{\boldsymbol{\kappa }}_{2}},{{\boldsymbol{\kappa }}_{1}}}} & {{\mathbf{I}}_{{{\boldsymbol{\kappa }}_{2}}}}  \\
\end{matrix} \right],
\end{equation}
\begin{equation}
{{\mathbf{i}}_{\boldsymbol{\kappa },g}}=\left[ \begin{matrix}
   {{\mathbf{i}}_{{{\boldsymbol{\kappa }}_{1}},g}}  \\
   {{\mathbf{i}}_{{{\boldsymbol{\kappa }}_{2}},g}}  \\
\end{matrix} \right],
\end{equation}
\begin{equation}
\mathbf{J}\left( \boldsymbol{\kappa } \right)=\left[ \begin{matrix}
   \mathbf{J}\left( {{\boldsymbol{\kappa }}_{1}} \right) & \mathbf{J}{{\left( {{\boldsymbol{\kappa }}_{2}},{{\boldsymbol{\kappa }}_{1}} \right)}^{T}}  \\
   \mathbf{J}\left( {{\boldsymbol{\kappa }}_{2}},{{\boldsymbol{\kappa }}_{1}} \right) & \mathbf{J}\left( {{\boldsymbol{\kappa }}_{2}} \right)  \\
\end{matrix} \right],
\end{equation}

\begin{subequations}
\begin{align}
& \mathbf{CR}{{\mathbf{B}}_{\boldsymbol{\kappa }}}=\mathbf{J}{{\left( \boldsymbol{\kappa } \right)}^{-1}}={{\left( {{\mathbf{I}}_{\boldsymbol{\kappa }}}-\frac{1}{{{i}_{g}}}{{\mathbf{i}}_{\boldsymbol{\kappa },g}}\mathbf{i}_{\boldsymbol{\kappa },g}^{T} \right)}^{-1}}, \\
\label{eq_l31}
& \mathbf{CR}{{\mathbf{B}}_{{{\boldsymbol{\kappa }}_{1}}}}=\mathbf{J}{{\left( {{\boldsymbol{\kappa }}_{1}} \right)}^{-1}}={{\left( {{\mathbf{I}}_{{{\boldsymbol{\kappa }}_{1}}}}-\frac{1}{{{i}_{g}}}{{\mathbf{i}}_{{{\boldsymbol{\kappa }}_{1}},g}}\mathbf{i}_{{{\boldsymbol{\kappa }}_{1}},g}^{T} \right)}^{-1}},
\end{align}
\end{subequations}
where the modified EFIM $\mathbf{J}\left( {{\boldsymbol{\kappa }}_{1}} \right)$ and $\mathbf{J}\left( {{\boldsymbol{\kappa }}_{2}} \right)$ compose the upper-left and lower-right matrix blocks of the original EFIM $\mathbf{J}\left( \boldsymbol{\kappa } \right)$. The elements of ${\mathbf{I}}_{\boldsymbol{\kappa }}$ and ${{\mathbf{i}}_{\boldsymbol{\kappa },g}}$, i.e. ${{\mathbf{I}}_{{{\boldsymbol{\kappa }}_{1}}}}$, ${{\mathbf{I}}_{{{\boldsymbol{\kappa }}_{1}},{{\boldsymbol{\kappa }}_{2}}}}$ and ${{\mathbf{i}}_{{{\boldsymbol{\kappa }}_{1}},g}}$, share similar definitions with elements in $(\ref{eq_l14})$. While $\mathbf{J}\left( {{\boldsymbol{\kappa }}_{1}} \right)$ is a shrunken version of $\mathbf{J}\left( \boldsymbol{\kappa } \right)$ with a size of $3 \times 3$, it provides a robust and closed-form CRB approximation for $\mathbf{J}\left( \boldsymbol{\kappa } \right)$ regarding the shared parameters ${{\boldsymbol{\kappa }}_{1}}$ in these two EFIMs, i.e. $CR{{B}_{{{\boldsymbol{\kappa }}_{1}}}}\left( {{\phi }_{o}} \right)\approx CR{{B}_{\boldsymbol{\kappa }}}\left( {{\phi }_{o}} \right)$.

From another aspect, the contour parameter ${{\boldsymbol{\kappa }}_{2}}$ is actually invariant for one specific extended target. We can exclude ${{\boldsymbol{\kappa }}_{2}}$ from $\boldsymbol{\kappa }$ to be estimated when the estimating accuracy is acceptable after several observations, or when ${{\boldsymbol{\kappa }}_{2}}$ is already known a prior to the BS. As such, the whole EFIM $\mathbf{J}\left( \boldsymbol{\kappa } \right)$ naturally degrades to the modified EFIM $\mathbf{J}\left( {{\boldsymbol{\kappa }}_{1}} \right)$. In the following paragraphs, we shall use $CRB$ to replace the expression of $CR{{B}_{{{\boldsymbol{\kappa }}_{1}}}}$ for simplicity. Finally, with pre-knowledge of the RCS distribution and the exact contour of the target, we have the following theorem:

\begin{theorem}
The CRB on ${{\phi}_{o}}$ can be expressed as $(\ref{eq_l16})$, ${{Z}_{1}}={{\pi }^{2}}\left( N_{r}^{2}-1 \right){{\cos }^{2}}\varphi /12$, ${{Z}_{2}}={{\left( 4\pi B/c \right)}^{2}}$, ${{X}_{k}}={{\left[ \boldsymbol{\sigma} _{k}^{T}\mathbf{m}\sin \left( {{\phi }_{o}}-\varphi  \right)-\boldsymbol{\varsigma} _{k}^{T}\mathbf{n}\cos \left( {{\phi }_{o}}-\varphi  \right) \right]}^{2}}$, ${{\mathbf{a}}_{k}}$ is the abbreviation for $\mathbf{a}\left( {{\phi }_{k}} \right)$, $B$ is the effective bandwidth of the ISAC system, ${{\boldsymbol{\sigma} }_{k}}={{\left[ \cos \left( {{u}_{k}} \right),...,\cos \left( Q{{u}_{k}} \right) \right]}^{T}}$ and ${{\boldsymbol{\varsigma} }_{k}}={{\left[ \sin \left( {{u}_{k}} \right),...,\sin \left( Q{{u}_{k}} \right) \right]}^{T}}$ are the cosine and sine harmonics.
\end{theorem}

\begin{proof}
\textit{Proof:} Please see Appendix I.{$\hfill\blacksquare$}
\end{proof}

\section{Joint Beamforming Design}
In this section, we aim to minimize the derived EST direction CRB under some practical constraints. First, we need to ensure the total power constraint at the BS and to guarantee the minimum SINR requirement for all CUs. Namely, the first two constraints are:

\begin{equation}
\label{eq_l17}
C1:\mathrm{tr}\left( {{\mathbf{R}}_{x}} \right)\le {{P}_{t}},
\end{equation}

\begin{equation}
\label{eq_l19}
C2:\left( 1+{{\Gamma }^{-1}} \right)\mathbf{h}_{c}^{H}{{\mathbf{R}}_{c}}{{\mathbf{h}}_{c}}\ge \mathbf{h}_{c}^{H}{{\mathbf{R}}_{x}}{{\mathbf{h}}_{c}}+\sigma _{c}^{2},
\end{equation}
where $\Gamma$ is the SINR threshold for all CUs, $\mathbf{R}_{c}=\mathbf{w}_{c}\mathbf{w}_{c}^{H}$ is the covariance matrix for the $c$-th CU.

Second, for EST, we need to ensure that every element on the LoS contour receives sufficient energy and reflects back valid echo signals, which is vital for further parameter estimation. Thus, we define a 3-dB beam coverage constraint $C3$, where the maximum received energy of each LoS contour element should cater to the corresponding minimum value. Note that in $(\ref{eq_l8})$, we discretize the LoS contour curve into $K$ subsections. Leveraging the max-min concept, $C3$ can be expressed in a discrete form as
\begin{align}
\label{eq_l21}
C3:2\min_{1 \leq k \leq K} \left( \mathbf{a}_{k}^{H}{\mathbf{R}}_{x}\mathbf{a}_{k} \right) - \max_{1 \leq k \leq K} \left( \mathbf{a}_{k}^{H}{\mathbf{R}}_{x}\mathbf{a}_{k} \right) \ge 0,
\end{align}
where the beampattern term $\mathbf{a}_{k}^{H}{{\mathbf{R}}_{x}}{{\mathbf{a}}_{k}}$ is utilized to describe the energy received by elements of the $k$-th LoS contour subsection. Based on above analysis, the ISAC beamforming optimization problem can be formulated as follows
\begin{align}
&\underset{\left\{ {{\mathbf{w}}_{c}} \right\}_{c=1}^{C}}{\mathop{\min }}\,CRB\left( {{\phi }_{o}} \right) \nonumber \\
&s.t.\text{ }C1,\text{ }C2,\text{ }C3.
\label{eq_l22}
\end{align}

Recalling the CRB formula in $(\ref{eq_l16})$, it is clear that the objective function $CRB\left( {{\phi }_{o}} \right)$ is non-convex owing to its fractional structure. To transform $CRB\left( {{\phi }_{o}} \right)$ into a convex expression, we utilize the Schur complement condition \cite{ref4} and introduce $K$ extra variables to rewrite the objective function as follows

\vspace{-\baselineskip} % 设置公式之前的垂直间距
\begin{align}
& \underset{\left\{ {{\mathbf{w}}_{c}} \right\}_{c=1}^C,\left\{ {{t}_{k}} \right\}_{k=1}^{K}}{\mathop{\min }}\,-\sum\limits_{k=1}^{K}{{{l}_{k}}{{t}_{k}}} \nonumber \\ 
s.t.& {\mathbf{P}_{k}}= \left[ \begin{matrix}
   \frac{{{\Vert {{{\mathbf{\dot{a}}}}_{k}} \Vert}^{2}}}{{{\Vert {{\mathbf{b}}_{k}} \Vert}^{2}}}\mathbf{a}_{k}^{H}{{\mathbf{R}}_{x}}{{\mathbf{a}}_{k}}+\mathbf{\dot{a}}_{k}^{H}{{\mathbf{R}}_{x}}{{{\mathbf{\dot{a}}}}_{k}}-{{t}_{k}} & \mathcal{R}\left(\mathbf{\dot{a}}_{k}^{H}{{\mathbf{R}}_{x}}{{\mathbf{a}}_{k}}\right)  \\
   \mathcal{R}\left(\mathbf{\dot{a}}_{k}^{H}{{\mathbf{R}}_{x}}{{\mathbf{a}}_{k}}\right) & \mathbf{a}_{k}^{H}{{\mathbf{R}}_{x}}{{\mathbf{a}}_{k}}  \\
\end{matrix} \right], \nonumber \\
&\mathbf{P}_k \ge 0, k=1,...,K, \nonumber \\
&{{\mathbf{R}}_{x}}=\sum\limits_{c=1}^{C}{{{\mathbf{R}}_{c}}},{{\mathbf{R}}_{c}}=\sum\limits_{c=1}^{C}\mathbf{w}_{c}\mathbf{w}_{c}^{H},c=1,...,C.
\label{eq_l23}
\end{align}
\vspace{-\baselineskip} % 设置公式之前的垂直间距

Replacing $(\ref{eq_l16})$ with the equivalent expression $(\ref{eq_l23})$, we observe that the modified problem is still non-convex due to the quadratic terms in $C1$, $C2$ and $C3$ of $(\ref{eq_l22})$, i.e. ${{\mathbf{R}}_{c}}={{\mathbf{w}}_{c}}\mathbf{w}_{c}^{H}$. To formulate a convex problem, one common practice is to employ the classical semidefinite programming (SDP) technique, replacing the original variable ${{\mathbf{w}}_{c}}$ in $(\ref{eq_l22})-(\ref{eq_l23})$ by ${{\mathbf{R}}_{c}}={{\mathbf{w}}_{c}}\mathbf{w}_{c}^{H}$ with rank-one constraint $\text{rank}\left( {{\mathbf{R}}_{c}} \right)=1$. Omitting this constraint leads to the final SDR problem as follows

\vspace{-\baselineskip} % 设置公式之前的垂直间距
\begin{align}
  &\underset{\left\{ {{\mathbf{R}}_{c}} \right\}_{c=1}^{C},\left\{ {{t}_{k}} \right\}_{k=1}^{K}}{\mathop{\min }}\,-\sum\limits_{k=1}^{K}{{{l}_{_{k}}}{{t}_{k}}} \nonumber \\ 
 s.t. & {\mathbf{P}_{k}}\ge 0,k=1,...,K, \nonumber \\ 
 & \left( 1+{{\Gamma }^{-1}} \right)\mathbf{h}_{c}^{H}{{\mathbf{R}}_{c}}{{\mathbf{h}}_{c}}\ge \mathbf{h}_{c}^{H}{{\mathbf{R}}_{x}}{{\mathbf{h}}_{c}}+\sigma _{c}^{2},c=1,...,C, \nonumber \\ 
 & \text{tr}\left( {{\mathbf{R}}_{x}} \right)\le {{P}_{t}}, \nonumber \\ 
 & 2\min_{1 \leq k \leq K} \left( \mathbf{a}_{k}^{H}{{\mathbf{R}}_{x}}{{\mathbf{a}}_{k}} \right)-\max_{1 \leq k \leq K} \left( \mathbf{a}_{k}^{H}{{\mathbf{R}}_{x}}{{\mathbf{a}}_{k}} \right)\ge 0,\nonumber \\ 
 & {{\mathbf{R}}_{x}}=\sum\limits_{c=1}^{C}{{{\mathbf{R}}_{c}}},{{\mathbf{R}}_{c}}\ge 0,c=1,...,C.
 \label{eq_l24}
\end{align}
\vspace{-\baselineskip}

It can be noted that the relaxed problem $(\ref{eq_l24})$ is a convex quadratic semidefinite
programming (QSDP) problem whose global optimum can be obtained by convex optimization toolboxes \cite{ref9}. Nevertheless, we can not guarantee $(\ref{eq_l24})$ to be a tight relaxation of the original SDP problem for all circumstances. In other words, the optimal solution $\left\{ {{\mathbf{R}}_{c}} \right\}_{c=1}^{C}$ for $(\ref{eq_l24})$ is not necessarily rank-one. Consequently, we present a rank-one solution extracting algorithm to obtain the valid beamformers, designed as follows

\begin{theorem}
If $\left\{ {{\mathbf{R}}_{c}} \right\}_{c=1}^{C}$ provide one optimal solution to $(\ref{eq_l24})$, then the following algorithm gives a solution $\left\{\tilde{\mathbf{w}_c}\right\}_{c=1}^C$ to $(\ref{eq_l24})$ without violating the corresponding power and SINR constraints $C1$ and $C2$

\vspace{-\baselineskip} % 设置公式之前的垂直间距
\begin{equation}
\begin{aligned}
&\mathbf{\tilde{u}}_{c}\in \text{span}\left( \mathbf{W}_c \right), \quad c=1,...,C, \\
&\boldsymbol{\tilde{\eta}}=\left[\gamma_{1}\sigma _{c}^{2},...\gamma_{C}\sigma _{c}^{2}\right]^{T}, \\
&\left[\mathbf{\tilde{F}}\right]_{i,c}=\begin{cases}
   \left|\mathbf{h}_{i}^{H}\mathbf{\tilde{u}}_{i}\right|^{2}, & i=c, \\
   \Gamma \left|\mathbf{h}_{c}^{H}\mathbf{\tilde{u}}_{i}\right|^{2}, & i\ne c,
\end{cases} \\
&\mathbf{\tilde{q}}=\mathbf{\tilde{F}}^{-1}\boldsymbol{\tilde{\eta}}=\left[\mathbf{\tilde{q}}_{1},...,\mathbf{\tilde{q}}_{C}\right]^{T}, \\
&\mathbf{\tilde{w}}_{c}=\sqrt{\mathbf{\tilde{q}}}_{c}\mathbf{\tilde{u}}_{c}, \quad c=1,...,C.
\end{aligned}
\label{eq_l25}
\end{equation}
\end{theorem}

%\vspace{-\baselineskip} % 设置公式之前的垂直间距
%\vspace{4pt} % 设置公式之前的垂直间距

\begin{proof}
\textit{Proof:} see Chapter 18 in \cite{ref10}.{$\hfill\blacksquare$}
\end{proof}

We note that the solution $\{{\mathbf{\tilde{w}}}_{c}\}_{c=1}^C$ from Theorem 2 is not necessarily guaranteed to satisfy the beam coverage constraint $C3$ in $(\ref{eq_l24})$. Consequently, ${{N}_{e}}$ random extraction epoches are required in $(\ref{eq_l25})$ to obtain a robust solution. Two potential circumstances could occur after the randomization process.

\begin{itemize}
    \item If there exists one or multiple feasible solutions in ${{N}_{e}}$ epochs, we select the epoch contributing to the minimum CRB as the final solution
    
    \vspace{-\baselineskip} % 设置公式之前的垂直间距
    \vspace{2pt} % 设置公式之前的垂直间距
    \begin{equation}
    \label{eq_l26}
    \left\{ \mathbf{\tilde{w}}_{c}^{opt} \right\}_{c=1}^{C}={\mathop{\min_{1\leq 
e \leq {N}_{e}} }}\,CRB\left( \left\{ {{{\mathbf{\tilde{w}}}}_{c,e}} \right\}_{c=1}^{C} \right),
    \end{equation}
    \vspace{-\baselineskip} % 设置公式之前的垂直间距
    
    where ${{\mathbf{\tilde{w}}}_{c,e}}$ refers to the extracted beamformer of the $c$-th CU in the $e$-th epoch.
    
    \item If there exists no feasible solution in all epochs, then the following cost function is applied to ${{N}_{e}}$ solutions to obtain a final beamformer with maximal value

    \vspace{-\baselineskip} % 设置公式之前的垂直间距
    \vspace{3pt} % 设置公式之前的垂直间距
    \begin{align}
   \mathcal{L}\left( \left\{ {{{\mathbf{\tilde{w}}}}_{c,e}} \right\}_{c=1}^{C} \right)&=2\min_{1\leq 
    k \leq K}  \left( \mathbf{a}_{k}^{H}\sum\limits_{c=1}^{C}{{{{\mathbf{\tilde{w}}}}_{c,e}}\mathbf{\tilde{w}}_{c,e}^{H}}{{\mathbf{a}}_{k}} \right) \nonumber \\ 
    \label{eq_l27}
     & -\max_{1\leq 
    k \leq K} \left( \mathbf{a}_{k}^{H}\sum\limits_{c=1}^{C}{{{{\mathbf{\tilde{w}}}}_{c,e}}\mathbf{\tilde{w}}_{c,e}^{H}}{{\mathbf{a}}_{k}} \right),
    \end{align}
    
    \begin{equation}
    \label{eq_l28}
     \left\{ \mathbf{\tilde{w}}_{c}^{opt} \right\}_{c=1}^{C}={\mathop{\max_{1\leq 
e \leq {N}_{e}} } }\,\mathcal{L}\left( \left\{ {{{\mathbf{\tilde{w}}}}_{c,e}} \right\}_{c=1}^{C} \right).
\end{equation}
%\vspace{-\baselineskip} % 设置公式之前的垂直间距
\end{itemize}

We outline the steps for calculating the precoding matrix $\mathbf{W}_c$ in Algorithm 1. Solving the QSDP problem contributes to the primary complex flops in Algorithm 1. Referred from \cite{ref11}, with a given solution accuracy $\varepsilon $, the worst case complexity to solve the QSDP problem with the primal-dual interior-point method is $\mathcal{O}\left( {{C}^{6.5}}N_{t}^{6.5}\log \left( 1/\varepsilon  \right) \right)$.
\begin{algorithm}[H]
\caption{ISAC Transmit Beamforming Design via SDR.}\label{alg:alg1}
\begin{algorithmic}
\STATE 
\STATE {\textsc{Input:}}
\STATE \hspace{0.5cm}$ \text{Communicaiton parameters:} P_t , \sigma_c^2 , \Gamma, \left\{\mathbf{h}_c\right\}_{c=1}^{C}  $
\STATE \hspace{0.5cm}$ \text{Sensing parameters:} \left\{l_k , \phi_k \right\}_{k=1}^K $
\STATE {\textsc{Output:}}
\STATE \hspace{0.5cm}$ \text{The transmit beamforming matrix } \mathbf{W}_{c}  $
\STATE {\textsc{Steps:}}
\begin{enumerate}
    \item {Compute the optimal solution $\left\{ {{\mathbf{R}}_{c}} \right\}_{c=1}^{C}$ by solving $(\ref{eq_l24})$ with cvx}
    \item {Obtain $\left\{ {{{\mathbf{\tilde{w}}}}_{c,e}} \right\}_{c=1,e=1}^{C,{{N}_{e}}}$ via $(\ref{eq_l25})$
    \item If there exist feasible solutions in ${{N}_{e}}$ epochs, compute $\left\{ \mathbf{\tilde{w}}_{c}^{opt} \right\}_{c=1}^{C}$ via $(\ref{eq_l26})$}
    \item {Else, compute $\left\{ \mathbf{\tilde{w}}_{c}^{opt} \right\}_{c=1}^{C}$ via $(\ref{eq_l27})-(\ref{eq_l28})$}
    \item {Obtain beamforming matrix $\mathbf{W}_{c}=\left[ \mathbf{\tilde{w}}_{1}^{opt},...,\mathbf{\tilde{w}}_{C}^{opt} \right]$}
\end{enumerate}
\end{algorithmic}
\label{alg1}
\end{algorithm}

\section{Numerical Results}
In this section, we perform numerical experiments to evaluate the performance of the proposed beamforming design, namely “CRB-min Design". If not specifically indicated, we consider an ISAC BS equipped with a uniform linear array (ULA) with ${{N}_{t}}=16$ transmit antennas and ${{N}_{r}}=16$ receive antennas. The transmit power is ${{P}_{t}}=0$ dBW. The noise power is set to $-80$ dBm. There exists $C=4$ downlink CUs located at ${{\phi }_{c}}=\left[ -{{60}^\circ},-{{35}^\circ},{{35}^\circ},{{60}^\circ} \right]$. The SINR threshold is set as $\Gamma =5$ dB. The EST is assumed to be located ${{d}_{o}}=27$ m away from the BS with a direction of ${{\phi }_{o}}={0}^\circ$ and orientation of $\varphi ={{0}^\circ}$. The LoS contour is divided into $K=8$ disjoint subsections, parameterized by $Q = 8$ TFS harmonics with $\mathbf{m}=\left[2.05,-0.002,0.5,0,0.056,0.001,-0.125,0.003\right]^T$ and $\mathbf{n}=\left[1.24,-0.001,0.335,-0.001,0.124,-0.001,0.018,0\right]^T$. The observation period is set as $t_s = 1\text{ s}$.

\begin{figure}[!t]
\setlength{\abovecaptionskip}{0.cm}
\setlength{\belowcaptionskip}{0.cm}
\centering
\includegraphics[width=3.5in]{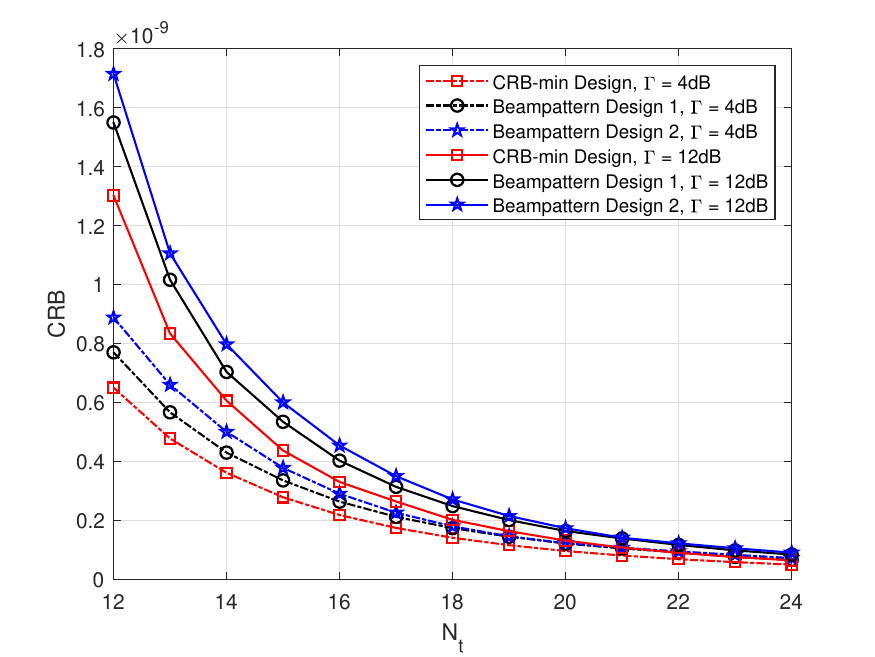}
\caption{CRBs as a function of transmit antenna numbers.}
\label{fig2}
\end{figure}

For comparsion, we select two beampattern-approaching precoder designs proposed in \cite{ref13} and \cite{ref11} as benchmarks, referred as “Beampattern Design 1” and “Beampattern Design 2”, respectively. The benchmark designs tend to allocate equal energy towards areas of interest, and we employ a main beam with a beamwidth of $\Delta ={10}^\circ$ in the simulation. The global direction grids are uniformly sampled in the range of $\left[ -{90}^\circ,{90}^\circ \right]$ with an interval of $1^\circ$.

The CRB value in terms of the numbers of transmit antennas and CUs are depicted in Fig. $\ref{fig2}$ and Fig. $\ref{fig3}$. An obvious CRB gap can be observed in Fig. $\ref{fig2}$ among all three beamforming schemes as the SINR threshold changes from 4 dB to 12 dB, which gradually narrows with the increase of antennas. A similar performance gap also exists in Fig. $\ref{fig3}$ where the CRB-min Design constantly yields the minimum CRB.

The fundamental difference of the aforementioned beamforming designs roots in the beam pattern preference. For the EST case, the beamforming design should consider both capturing the target direction and covering the whole LoS contour. The estimation of one specific direction prefers a sharp beam, whereas the effective reception of echo signals reflected from the LoS contour requires a sufficiently wide beam. The considered two benchmark beampattern designs only consider the contour coverage requirement by allocating equal energy to a given angle range, and yet ignores the direction capturing requirement. In contrast, our proposed design achieves a balance between the contradictory sharp and wide beam requests, generating a desirable transmit beampattern with the minimum CRB.

\begin{figure}[!t]
\setlength{\abovecaptionskip}{0.cm}
\setlength{\belowcaptionskip}{0.cm}
\centering
\includegraphics[width=3.5in]{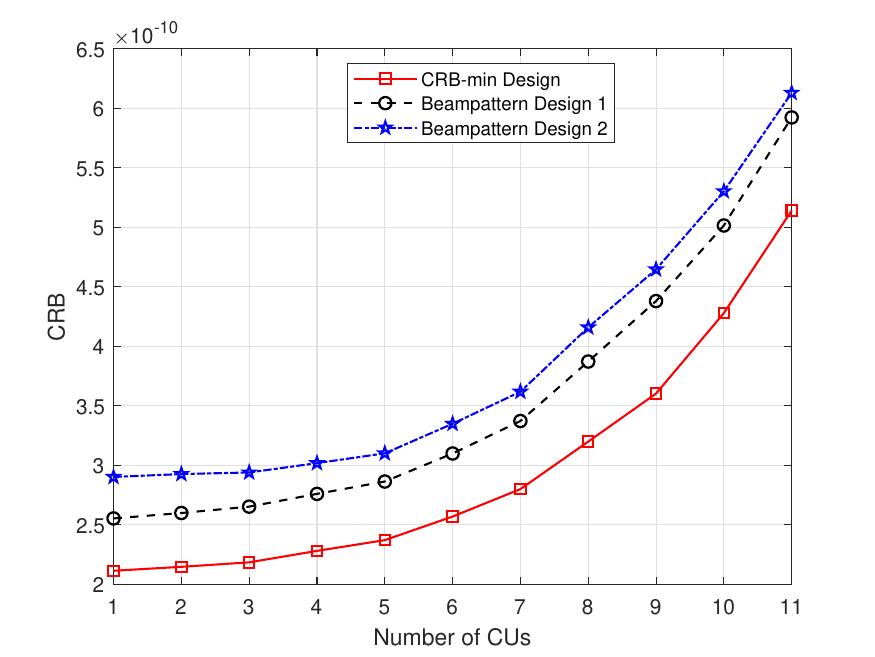}
\caption{CRBs as a function of CU numbers.}
\label{fig3}
\end{figure}

\section{Conclusion}
This paper considers the transmit beamforming design problem in an MU-MIMO ISAC system, where one BS communicates with multiple CUs and senses one extended target at the same time. Building on the basis of TFS target contour modeling, we first derive a closed-form CRB on the central direction of the EST. Then, the precoders at the BS are designed to minimize the CRB on the target direction while satisfying the SINR and power constraints for communication, as well as the beam coverage constraint for sensing. To solve the non-convex optimization problem, we employ SDR technique and propose rank-one solution extraction scheme for non-tight relaxation scenarios. Numerical results verify the robustness of our CRB-min beamforming design which outperforms benchmark designs with lower CRBs.

{\appendices
\section*{Appendix I \\ Proof of the Theorem 1}
\subsection*{A.\hspace{5pt}General Derivation of $\mathbf{J}(\boldsymbol{\kappa}_{1})$}
We commence by defining ${{\mathbf{d}}_{k}}=\mathbf{b}\left( {{\phi }_{k}} \right){{\mathbf{a}}^{H}}\left( {{\phi }_{k}} \right)\mathbf{x}\left( t-{2{{d}_{k}}} / {c} \right)$, the echo signal in $(\ref{eq_l8})$ can be correspondingly expressed as $\mathbf{e}\left( t \right)=g\sum\nolimits_{k=1}^{K}{\sqrt{{l}_{k}}}{{\alpha }_{k}}{\mathbf{d}_{k}}$. Under the definition of ${{\mathbf{I}}_{{{\boldsymbol{\kappa }}_{1}}}}=-\mathbb{E}\left[ \Delta _{{{\boldsymbol{\kappa }}_{1}}}^{{{\boldsymbol{\kappa }}_{1}}}\log p\left( {{\mathbf{y}}_{s}}|\boldsymbol{\xi } \right) \right]$, we have
\begin{align}
   {{\mathbf{I}}_{{{\boldsymbol{\kappa }}_{1}}}}&=\mathbb{E}\left[ \frac{2{{g}^{2}}}{\sigma _{s}^{2}}\mathcal{R}\int_{{{t}_{s}}}{\sum\limits_{{{k}_{1}}=1}^{K}{\sum\limits_{{{k}_{2}}=1}^{K}{\sqrt{{{l}_{{{k}_{1}}}}{{l}_{{{k}_{2}}}}}\alpha _{{{k}_{1}}}^{*}{{\alpha }_{{{k}_{2}}}}\frac{\partial \mathbf{d}_{{{k}_{1}}}^{H}}{\partial {{\boldsymbol{\kappa }}_{1}}}\frac{\partial {{\mathbf{d}}_{{{k}_{2}}}}}{\partial \boldsymbol{\kappa }_{1}^{T}}\text{d}t}}} \right] \nonumber \\ 
 & =\frac{2{{g}^{2}}}{\sigma _{s}^{2}}\mathcal{R}\int_{{{t}_{s}}}{\sum\limits_{{{k}_{1}}=1}^{K}{\sum\limits_{{{k}_{2}}=1}^{K}{\mathbb{E}\left[ \alpha _{{{k}_{1}}}^{*}{{\alpha }_{{{k}_{2}}}} \right]\sqrt{{{l}_{{{k}_{1}}}}{{l}_{{{k}_{2}}}}}\frac{\partial \mathbf{d}_{{{k}_{1}}}^{H}}{\partial {{\boldsymbol{\kappa }}_{1}}}\frac{\partial {{\mathbf{d}}_{{{k}_{2}}}}}{\partial \boldsymbol{\kappa }_{1}^{T}}}}}\text{d}t \nonumber\\ 
 & =\frac{2{{g}^{2}}}{\sigma _{s}^{2}}\sum\nolimits_{k=1}^{K}{{{l}_{k}}}\mathcal{R}\int_{{{t}_{s}}}{\frac{\partial \mathbf{d}_{k}^{H}}{\partial {{\boldsymbol{\kappa }}_{1}}}\frac{\partial {{\mathbf{d}}_{k}}}{\partial \boldsymbol{\kappa }_{1}^{T}}\text{d}t},
\end{align}
where the third equation holds true since ${{\alpha }_{k}} \sim \mathcal{CN}\left( 0,1 \right)$, $\mathbb{E}\left[ \alpha _{{{k}_{1}}}^{*}{{\alpha }_{{{k}_{2}}}} \right]=1$ for ${{k}_{1}}={{k}_{2}}$ and $\mathbb{E}\left[ \alpha _{{{k}_{1}}}^{*}{{\alpha }_{{{k}_{2}}}} \right]=0$ for ${{k}_{1}}\ne {{k}_{2}}$. We notice that the calculation of $\partial \mathbf{d}_{k}^{H}/\partial {{\boldsymbol{\kappa }}_{1}}$ is complicated since ${{\mathbf{d}}_{k}}$ is directly linked with intermediate variables ${{\boldsymbol{\Theta }}_{k}}=\left[ {{d}_{k}},{{\phi }_{k}} \right]^T$ which depend on the contour parameters $\left\{ {{a}_{q}},{{b}_{q}} \right\}_{q=1}^{Q}$. Thus, following the chain rule $ \partial\mathbf{d}_{k}^{H}/ \partial{{\boldsymbol{\kappa }}_{1}}=\frac{\partial \mathbf{d}_{k}^{H}}{\partial {{\boldsymbol{\Theta }}_{k}}}\frac{\partial \boldsymbol{\Theta }_{k}^{T}}{\partial {{\boldsymbol{\kappa }}_{1}}}$, we can rewrite the FIM as follows
\begin{equation}
\label{eq_l111}
    {{\mathbf{I}}_{{{\boldsymbol{\kappa }}_{1}}}}=\frac{2{{g}^{2}}}{\sigma _{s}^{2}}\sum\nolimits_{k=1}^{K}{{{l}_{k}}}\frac{\partial \boldsymbol{\Theta }_{k}^{T}}{\partial {{\boldsymbol{\kappa }}_{1}}}\left(\mathcal{R}\int_{{{t}_{s}}}{\frac{\partial \mathbf{d}_{k}^{H}}{\partial {{\boldsymbol{\Theta }}_{k}}}\frac{\partial {{\mathbf{d}}_{k}}}{\partial \boldsymbol{\Theta }_{k}^{T}}\text{d}t} \right)\frac{\partial {{\boldsymbol{\Theta }}_{k}}}{\partial \boldsymbol{\kappa }_{1}^{T}}.
\end{equation}

With the formulas in Appendix I-B to Appendix I-E, $(\ref{eq_l111})$ can be further written as

\begin{align}
{{\mathbf{I}}_{{{\boldsymbol{\kappa }}_{1}}}}=&\frac{2{{g}^{2}}N_{r}}{\sigma _{s}^{2}}\sum\limits_{k=1}^{K}{{{l}_{k}}\mathbf{a}_{k}^{H}{{\mathbf{R}}_{x}}{{\mathbf{a}}_{k}}} \nonumber\\
&\left[ \begin{matrix}
{{\mathbf{\mu }}_{k}} & {{\mathbf{\eta }}_{k}}  \\
\end{matrix} \right]\left[ \begin{matrix}
{{Z}_{2}} & 0  \\
0 & t_s({{Z}_{1}}+\frac{\mathbf{\dot{a}}_{k}^{H}{{\mathbf{R}}_{x}}{{{\mathbf{\dot{a}}}}_{k}}}{\mathbf{a}_{k}^{H}{{\mathbf{R}}_{x}}{{\mathbf{a}}_{k}}})  \\
\end{matrix} \right]\left[ \begin{matrix}
\mathbf{\mu }_{k}^{T}  \\
\mathbf{\eta }_{k}^{T}  \\
\end{matrix} \right] \nonumber\\ 
 =&\frac{2{{g}^{2}}N_{r}}{\sigma _{s}^{2}}\sum\limits_{k=1}^{K}{{{l}_{k}}\mathbf{a}_{k}^{H}{{\mathbf{R}}_{x}}{{\mathbf{a}}_{k}}} \nonumber \\
&\left[ {{Z}_{2}}{{\mathbf{\mu }}_{k}}\mathbf{\mu }_{k}^{T}+\left( {{Z}_{1}}+\frac{\mathbf{\dot{a}}_{k}^{H}{{\mathbf{R}}_{x}}{{{\mathbf{\dot{a}}}}_{k}}}{\mathbf{a}_{k}^{H}{{\mathbf{R}}_{x}}{{\mathbf{a}}_{k}}} \right){{t_s\mathbf{\eta }}_{k}}\mathbf{\eta }_{k}^{T} \right].
\label{eq_l29}
\end{align}

Similar to the calculation steps in $(\ref{eq_l29})$, we obtain ${{i}_{g}}$ and ${{\mathbf{i}}_{{{\boldsymbol{\kappa }}_{1}},g}}$ as following
\begin{equation}
{{i}_{g}}=\frac{2{{N}_{r}}t_s}{\sigma _{s}^{2}}\sum\limits_{k=1}^{K}{{{l}_{k}}\mathbf{a}_{k}^{H}{{\mathbf{R}}_{x}}{{\mathbf{a}}_{k}}},
\end{equation}

\begin{align}
{{\mathbf{i}}_{{{\boldsymbol{\kappa }}_{1}},g}}  &=\frac{2g}{\sigma _{s}^{2}}\sum\limits_{k=1}^{K}{{{l}_{k}}\frac{\partial \boldsymbol{\Theta }_{k}^{T}}{\partial {{\boldsymbol{\kappa }}_{1}}}}\left( \mathcal{R}\int_{{{t}_{s}}}{\frac{\partial \mathbf{d}_{k}^{H}}{\partial {{\boldsymbol{\Theta }}_{k}}}{{\mathbf{d}}_{k}}\text{d}t} \right) \nonumber \\
&=\left[0 \quad  
   \frac{g{{N}_{r}}t_s}{\sigma _{s}^{2}}\sum\limits_{k=1}^{K}{{{l}_{k}}\left( \mathbf{\dot{a}}_{k}^{H}{{\mathbf{R}}_{x}}{{\mathbf{a}}_{k}}+\mathbf{a}_{k}^{H}{{\mathbf{R}}_{x}}{{{\mathbf{\dot{a}}}}_{k}} \right)} \quad 0\right]^T.
\label{eq_l30}
\end{align}

Finally, combine formulas $(\ref{eq_l29})-(\ref{eq_l30})$ with formula $(\ref{eq_l31})$, we can get the closed-form CRBs in $(\ref{eq_l5})-(\ref{eq_l7})$.

\subsection*{B.\hspace{5pt}Derivation of ${\partial \boldsymbol{\Theta }_{k}^{T}}/{\partial {{\boldsymbol{\kappa }}_{1}}}$}
We decompose $\partial \boldsymbol{\Theta }_{k}^{T}/\partial {{\boldsymbol{\kappa }}_{1}}=\left[ \partial {{d}_{k}}/\partial {{\boldsymbol{\kappa }}_{1}} \quad \partial {{\phi }_{k}}/\partial {{\boldsymbol{\kappa }}_{1}} \right]$ as following

\begin{align}
{{\mathbf{\mu }}_{k}}&=\frac{\partial {{{d}}_{k}}}{\partial {{\boldsymbol{\kappa }}_{1}}}=\left[ 
   \frac{\partial {{{d}}_{k}}}{\partial {{d}_{o}}} \quad \frac{\partial {{{d}}_{k}}}{\partial {{\phi }_{o}}} \quad
   \frac{\partial {{{d}}_{k}}}{\partial \varphi }  \right]^T \nonumber \\
 &=\frac{1}{{{d}_{k}}}\left[ \begin{matrix}
   {{d}_{o}}+d_{o}^{-1}\mathbf{r}_{k}^{T}{{\mathbf{V}}^{T}}{{\mathbf{p}}_{k}}  \\
   \mathbf{r}_{k}^{T}{{\mathbf{V}}^{T}}{{\mathbf{p}}_{\bot ,k}}  \\
   -\mathbf{r}_{k}^{T}{{\mathbf{V}}^{T}}{{\mathbf{p}}_{\bot ,k}}  \\
\end{matrix} \right]\approx \left[ \begin{matrix} 1  \\
   \mathbf{r}_{k}^{T}{{\mathbf{V}}^{T}}{{\mathbf{p}}_{\bot }}/{{d}_{o}}  \\
   -\mathbf{r}_{k}^{T}{{\mathbf{V}}^{T}}{{\mathbf{p}}_{\bot }}/{{d}_{o}}  \\
\end{matrix} \right],
\end{align}

\begin{align}
{{\mathbf{\eta }}_{k}}&=\frac{\partial {{\phi }_{k}}}{\partial {{\boldsymbol{\kappa }}_{1}}}=\left[ 
   \frac{\partial {{\phi }_{k}}}{\partial {{d}_{o}}}  \quad
   \frac{\partial {{\phi }_{k}}}{\partial {{\phi }_{o}}}  \quad
   \frac{\partial {{\phi }_{k}}}{\partial \varphi }   \right]^T \nonumber \\
   &=\frac{1}{d_{k}^{2}}\left[ \begin{matrix}
   d_{o}^{-1}\mathbf{r}_{k}^{T}{{\mathbf{V}}^{T}}{{\mathbf{p}}_{\bot ,k}}  \\
   d_{o}^{2}+\mathbf{r}_{k}^{T}{{\mathbf{V}}^{T}}{{\mathbf{p}}_{k}}  \\
   \mathbf{r}_{k}^{T}{{\mathbf{V}}^{T}}{{\mathbf{p}}_{k}}  \\
\end{matrix} \right]\approx \left[ \begin{matrix}
   0  \\
   1  \\
   0  \\
\end{matrix} \right],
\end{align}
where parameters $\left\{ {{d}_{k}},{{\phi }_{k}},{{\mathbf{r}}_{k}}\mathbf{,}{{\mathbf{p}}_{k}} \right\}$ are respectively the range, direction, local position and global position of the $k$-th contour section, ${{\mathbf{p}}_{\bot ,k}}=\left( \begin{matrix} 0 & -1  \\ 1 & 0  \\ \end{matrix} \right){{\mathbf{p}}_{k}}$, ${{\mathbf{p}}_{\bot }}=\left( \begin{matrix} 0 & -1  \\ 1 & 0  \\ \end{matrix} \right){{\mathbf{p}}_{o}}$.

\subsection*{C.\hspace{5pt}Derivation of $\mathcal{R}\int_{{{t}_{s}}}{\frac{\partial \mathbf{d}_{k}^{H}}{\partial {{\boldsymbol{\Theta }}_{k}}}\frac{\partial {{\mathbf{d}}_{k}}}{\partial \boldsymbol{\Theta }_{k}^{T}}\text{d}t}$}
We start with the calculation of $\frac{\partial {{\mathbf{d}}_{k}}}{\partial \boldsymbol{\Theta }_{k}^{T}}=\left[ \frac{\partial {{\mathbf{d}}_{k}}}{\partial {{d}_{k}}}\text{ }\frac{\partial {{\mathbf{d}}_{k}}}{\partial {{\phi }_{k}}} \right]$
\begin{equation}
\frac{\partial {{\mathbf{d}}_{k}}}{\partial {{d}_{k}}}=-\frac{2}{c}{{\mathbf{b}}_{k}}\mathbf{a}_{k}^{H}\mathbf{\dot{x}}\left( t-\frac{2{{d}_{k}}}{c} \right),
\end{equation}
\begin{equation}
\frac{\partial {{\mathbf{d}}_{k}}}{\partial {{\phi }_{k}}}=\left( {{{\mathbf{\dot{b}}}}_{k}}\mathbf{a}_{k}^{H}+{{\mathbf{b}}_{k}}\mathbf{\dot{a}}_{k}^{H} \right)\mathbf{x}\left( t-\frac{2{{d}_{k}}}{c} \right),
\end{equation}
where $\mathbf{\dot{x}}\left( t \right)=\partial \mathbf{x}\left( t \right)/\partial t$, ${{\mathbf{\dot{a}}}_{k}}=\partial {{\mathbf{a}}_{k}}/\partial \phi $ and ${{\mathbf{\dot{b}}}_{k}}=\partial {{\mathbf{b}}_{k}}/\partial \phi $.

With the extra identities listed in Appendix I-E, we further derive $\mathcal{R}\int_{{{t}_{s}}}{\frac{\partial \mathbf{d}_{k}^{H}}{\partial {{\boldsymbol{\Theta }}_{k}}}\frac{\partial {{\mathbf{d}}_{k}}}{\partial \boldsymbol{\Theta }_{k}^{T}}\text{d}t}$ as
\begin{equation}
\mathcal{R}\int_{{{t}_{s}}}{\frac{\partial \mathbf{d}_{k}^{H}}{\partial {{d}_{k}}}\frac{\partial {{\mathbf{d}}_{k}}}{\partial {{d}_{k}}}\text{d}t}={{N}_{r}}{{Z}_{2}}\mathbf{a}_{k}^{H}{\mathbf{R}_{x}}{{\mathbf{a}}_{k}},
\end{equation}
\begin{equation}
\mathcal{R}\int_{{{t}_{s}}}{\frac{\partial \mathbf{d}_{k}^{H}}{\partial {{\phi }_{k}}}\frac{\partial {{\mathbf{d}}_{k}}}{\partial {{\phi }_{k}}}\text{d}t}={{N}_{r}t_s}\left( {{Z}_{1}}\mathbf{a}_{k}^{H}{\mathbf{R}_{x}}{{\mathbf{a}}_{k}}+\mathbf{\dot{a}}_{k}^{H}{\mathbf{R}_{x}}{{{\mathbf{\dot{a}}}}_{k}} \right),
\end{equation}
\begin{equation}
\mathcal{R}\int_{{{t}_{s}}}{\frac{\partial \mathbf{d}_{k}^{H}}{\partial {{d}_{k}}}\frac{\partial {{\mathbf{d}}_{k}}}{\partial {{\phi }_{k}}}\text{d}t}=0.
\end{equation}

The complete matrix writes as
\begin{equation}
\mathcal{R}\int_{{{t}_{s}}}{\frac{\partial \mathbf{d}_{k}^{H}}{\partial {{\boldsymbol{\Theta }}_{k}}}\frac{\partial {{\mathbf{d}}_{k}}}{\partial \boldsymbol{\Theta }_{k}^{T}}\text{d}t}={{N}_{r}}\mathbf{a}_{k}^{H}{{R}_{x}}{{\mathbf{a}}_{k}}\left[ \begin{matrix}
   {{Z}_{2}} & 0  \\
   0 & t_s({{Z}_{1}}+\frac{\mathbf{\dot{a}}_{k}^{H}{{\mathbf{R}}_{x}}{{{\mathbf{\dot{a}}}}_{k}}}{\mathbf{a}_{k}^{H}{\mathbf{R}_{x}}{{\mathbf{a}}_{k}}})  \\
\end{matrix} \right].
\end{equation}

\subsection*{D.\hspace{5pt}Derivation of $\mathcal{R}\int_{{{t}_{s}}}{\frac{\partial \mathbf{d}_{k}^{H}}{\partial {{\boldsymbol{\Theta }}_{k}}}{{\mathbf{d}}_{k}}\text{d}t}$}
Use the identities in Appendix I-E, we make following derivation
\begin{equation}
\mathcal{R}\int_{{{t}_{s}}}{\frac{\partial \mathbf{d}_{k}^{H}}{\partial {{d}_{k}}}{{\mathbf{d}}_{k}}\text{d}t}=0,
\end{equation}

\begin{align}
\mathcal{R}\int_{{{t}_{s}}}{\frac{\partial \mathbf{d}_{k}^{H}}{\partial {{\phi }_{k}}}{{\mathbf{d}}_{k}}\text{d}t}&={{N}_{r}t_s}\mathcal{R}\left( \mathbf{\dot{a}}_{k}^{H}{{\mathbf{R}}_{x}}{{\mathbf{a}}_{k}} \right) \nonumber \\
&=\frac{{{N}_{r}t_s}}{2}\left( \mathbf{\dot{a}}_{k}^{H}{{\mathbf{R}}_{x}}{{\mathbf{a}}_{k}}+\mathbf{a}_{k}^{H}{{\mathbf{R}}_{x}}{{{\mathbf{\dot{a}}}}_{k}} \right).
\end{align}

The complete matrix writes as
\begin{equation}
\mathcal{R}\int_{{{t}_{s}}}{\frac{\partial \mathbf{d}_{k}^{H}}{\partial {{\boldsymbol{\Theta }}_{k}}}{{\mathbf{d}}_{k}}\text{d}t}=\left[ \begin{matrix}
   0 & \frac{{{N}_{r}}t_s}{2}\left( \mathbf{\dot{a}}_{k}^{H}{{\mathbf{R}}_{x}}{{\mathbf{a}}_{k}}+\mathbf{a}_{k}^{H}{{\mathbf{R}}_{x}}{{{\mathbf{\dot{a}}}}_{k}} \right)  \\
\end{matrix} \right].
\end{equation}

\subsection*{E.\hspace{5pt}Related Identities}
To facilitate the calculations, we take the center of the ULA as the reference point with zero phase. Thus, the transmit and receive antenna response can be presented as 
\begin{align}
&{{\mathbf{a}}_{k}}=\exp \left( j\pi \left( {{N}_{t}}-1 \right)/2\sin {{\phi }_{k}} \right) \nonumber\\
&\hspace*{.8cm}{{\left[ 1\text{ exp}\left( -j\pi \sin {{\phi }_{k}} \right)\text{ }...\text{ exp}\left( -j\left( {{N}_{t}}-1 \right)\pi \sin {{\phi }_{k}} \right) \right]}^{T}},
\end{align}
\setlength\abovedisplayskip{1pt}
\setlength\belowdisplayskip{1pt}
\begin{align}
&{{\mathbf{b}}_{k}}=\exp \left( j\pi \left( {{N}_{r}}-1 \right)/2\sin {{\phi }_{k}} \right) \nonumber\\
&\hspace*{.8cm}{{\left[ 1\text{ exp}\left( -j\pi \sin {{\phi }_{k}} \right)\text{ }...\text{ exp}\left( -j\left( {{N}_{r}}-1 \right)\pi \sin {{\phi }_{k}} \right) \right]}^{T}}.
 \end{align}

With basic operations, we can further calculate that ${{\left\| {{\mathbf{b}}_{k}} \right\|}^{2}}={{N}_{r}}$, and get
\begin{equation}
{{\mathbf{\dot{a}}}_{k}}=j\pi \cos \left( {{\phi }_{k}} \right)\text{diag}\left( \frac{{{N}_{t}}-1}{2},...,-\frac{{{N}_{t}}-1}{2} \right){{\mathbf{a}}_{k}},
\end{equation}
\begin{equation}
{{\mathbf{\dot{b}}}_{k}}=j\pi \cos \left( {{\phi }_{k}} \right)\text{diag}\left( \frac{{{N}_{r}}-1}{2},...,-\frac{{{N}_{r}}-1}{2} \right){{\mathbf{b}}_{k}},
\end{equation}
\begin{equation}
{{\left\| {{{\mathbf{\dot{b}}}}_{k}} \right\|}^{2}}=\frac{{{\cos }^{2}}\left( {{\phi }_{k}} \right){{\pi }^{2}}\left( N_{r}^{3}-{{N}_{r}} \right)}{12}={{N}_{r}}{{Z}_{1}}.
\end{equation}

Next, with the basic assumption of $\mathbb{E} \left[\mathbf{c}(t) \mathbf{c}^H (t) \right] = \mathbf{I}_{C}$, it is clear that $\int_{t_s} \mathbf{c}(t) \mathbf{c}^H (t) \text{d}t = t_s\mathbf{I}_{C}$. Since we are considering a sufficient long observation period $t_s$, the following Fourier transform derivation holds when the integration is calculated within $t_s$
\begin{align}
  & \int_{t_s}{c_{{{c}_{1}}}^{*}\left( t-\tau  \right){{{\dot{c}}}_{{{c}_{2}}}}\left( t-\tau  \right)\text{d}t} \nonumber\\ 
 & \approx\int_{-\infty }^{+\infty }{{{\left[ {{C}_{{{c}_{1}}}}\left( f \right){{e}^{-j2\pi \tau f}} \right]}^{*}}j2\pi f{{C}_{{{c}_{2}}}}\left( f \right){{e}^{-j2\pi \tau f}}\text{d}f} \nonumber\\ 
 & =j2\pi \int_{-\infty }^{+\infty }{fC_{{{c}_{1}}}^{*}\left( f \right){{C}_{{{c}_{2}}}}\left( f \right)\text{d}f}=0, c_{1,2}=1,...,C.
 \label{eq_l33}
\end{align}

For ${{c}_{1}}\ne {{c}_{2}}$, $(\ref{eq_l33})$ holds since $\int_{-\infty }^{+\infty }{fC_{{{c}_{1}}}^{*}\left( f \right){{C}_{{{c}_{2}}}}\left( f \right)\text{d}f}=0$ given the irrelevance between ${{C}_{{{c}_{1}}}}\left( f \right)$ and ${{C}_{{{c}_{2}}}}\left( f \right)$. For ${{c}_{1}}={{c}_{2}}=c$, the integral term $\int_{-\infty }^{+\infty }{f{{\left| {{C}_{c}}\left( f \right) \right|}^{2}}\text{d}f}$ can be regarded as the mass center of the signal spectrum ${{\left| {{C}_{c}}\left( f \right) \right|}^{2}}$. Since such center can be shifted arbitrarily in the frequency domain, we can locate it at a specific point with zero value such that $(\ref{eq_l33})$ still holds. Further, we have

\begin{align}
\label{eq_l34}
&\int_{t_s}{\dot{c}_{{{c}_{1}}}^{*}\left( t-\tau  \right){{{\dot{c}}}_{{{c}_{2}}}}\left( t-\tau  \right)\text{d}t} \nonumber\\
&\approx{{\left( 2\pi  \right)}^{2}}\int_{-\infty }^{+\infty }{{{f}^{2}}C_{{{c}_{1}}}^{*}\left( f \right){{C}_{{{c}_{2}}}}\left( f \right)\text{d}f} \nonumber\\
&=\left\{ \begin{matrix}{{\left( 2\pi B \right)}^{2}},{{c}_{1}}={{c}_{2}}  \\ 0,{{c}_{1}}\ne {{c}_{2}} \\
\end{matrix} \right., c_{1,2}=1,...,C.
\end{align}

Similar with $(\ref{eq_l33})$, for ${{c}_{1}}\ne {{c}_{2}}$, ${{C}_{{{c}_{1}}}}\left( f \right)$ is irrelevant with ${{C}_{{{c}_{2}}}}\left( f \right)$ so that $(\ref{eq_l34})$ holds. For ${{c}_{1}}= {{c}_{2}}=c$, $(\ref{eq_l34})$ follows the definition of $B=\left(\int_{-\infty }^{+\infty }{{{f}^{2}}{{\left| {{S}_{c}}\left( f \right) \right|}^{2}}\text{d}f}\right)^{1/2}$ which is the effective bandwidth of the ISAC system.}

\input{Reference1.bbl}

\end{document}